\documentclass[10pt]{article}
\usepackage[utf8]{inputenc}
\usepackage{amsmath,amsthm}
\usepackage{lineno}
\usepackage{hhline}
\usepackage[hidelinks]{hyperref}
\usepackage{listings}
\usepackage{color}
\definecolor{mygray}{rgb}{0.4,0.4,0.4}

\usepackage{amssymb}
\usepackage[table]{xcolor}
\usepackage{color, colortbl}
\definecolor{Gray}{gray}{0.9}
\definecolor{LightCyan}{rgb}{0.88,1,1}
\usepackage{multirow}
\usepackage{multicol}
\usepackage{todonotes}

\usepackage{psfrag}
\usepackage{placeins}
\usepackage{enumitem}
\usepackage{algorithm,algcompatible,algpseudocode}   
\modulolinenumbers[5]

\usepackage{caption}

\DeclareCaptionFormat{algor}{%
  \hrulefill\par\offinterlineskip\vskip1pt%
    \textbf{#1#2}#3\offinterlineskip\hrulefill}
\DeclareCaptionStyle{algori}{singlelinecheck=off,format=algor,labelsep=space}
\newtheorem{lemma}{Lemma}

\newtheorem{defn}{Definition}

\newcommand{\E}{\mathcal E}
\newcommand{\V}{\mathcal V}
\newcommand{\G}{\mathcal G}

\newcommand{\N}{\mathbb N}
\newcommand{\map}{\mathrm{map}}
\newcommand{\Rmap}{D_{\mathrm{map}}}
\newcommand{\tRmap}{C_{\mathrm{map}}}
\newcommand{\Image}{\mathrm{Image}}
\newcommand{\Emap}{E_{\mathrm{map}}}
\newcommand{\ERmap}{ER_{\mathrm{map}}}

\input{listings-modelica.cfg}

\title{Connected Components in Undirected Set--Based Graphs. Applications in Object--Oriented Model Manipulation.}
\author{Ernesto Kofman, Denise Marzorati and Joaquín Fernández\\
CIFASIS-CONICET, FCEIA-UNR, Argentina}
\date{}

\begin{document}
\maketitle



\begin{abstract}
This work introduces a novel algorithm for finding the connected components of a graph where the vertices and edges are grouped into sets defining a \emph{Set--Based Graph}. The algorithm, under certain restrictions on those sets, has the remarkable property of achieving constant computational costs with the number of vertices and edges. The mentioned restrictions are related to the possibility of representing the sets of vertices by intension and the sets of edges using some particular type of maps. While these restrictions can result strong in a general context, they are usually satisfied in the problem of transforming connections into equations in object oriented models, which is the main application of the proposed algorithm.

Besides describing the new algorithm and studying its computational cost, the work describes its prototype implementation and shows its application in different examples.
\end{abstract}



\section{Introduction}
Finding the connected components of an undirected graphs is a classic problem of Graph Theory that is employed in several application domains. Simple algorithms that solve this problem in linear time with the number of vertices have been known since several decades ago \cite{hopcroft1973algorithm}. Also, parallel algorithms that can solve the problem in logarithmic time have been known for long time \cite{hirschberg1979computing}.

One particular problem that requires finding the connected components of a graph is that of flattening the equations of object oriented models \cite{fritzson1998modelica}, which is part of the first stage of the compilation process. There, different sub-models are related by \emph{connectors} and the connections must be replaced by equations where sum of all connected variables of certain type must be zero. While solving the problem in linear time may be affordable in several situations, there are models that result of the coupling of thousands of small sub-models where the cost can become prohibitive. Moreover, even if the problem is solved in a reasonable amount of time, the resulting system of equations can be so large that it is intractable by the subsequent stages of the compilation process. 

Fortunately, large models often contain several repetitive connections that are the result of using \texttt{for} statements and this is a fact that can be exploited to reduce the computational cost of the different compilation stages \cite{arzt2014towards, schuchart2015exploiting, casella2015simulation, bergero2015efficient, stavaaker2015contributions, braun2017solving, qin2016efficient, agosta2019towards, schweiger2020modeling}. However, the possibility of exploiting the presence of repetitive or regular structures at each stage requires that the previous stages had kept a compact representation. 
While there are some experimental implementations that in some particular cases can keep a compact representation during the whole compilation process \cite{bergero2015efficient}, there is not yet a general solution. 

Regarding the flattening stage, a general solution would require to find the sets of connected connectors which may be part of multidimensional arrays, solving the problem without actually expanding those arrays into individual connectors. This problem is equivalent to find the connected components of an undirected graph while keeping some sets of vertices and edges grouped together, which constitutes the main goal of the present work.

The problem of manipulating large graphs grouping vertices and edges into sets to produce compact systems of equations was recently proposed with the introduction of \emph{Set--Based Graphs} \cite{ZFK19}. There, a compact solution for the problems of maximum matching and finding strongly connected components in directed graph for equation sorting was proposed and implemented as part of the prototype ModelicaCC compiler \cite{bergero2015efficient}.

In this work, we use the same tool (Set-Based Graphs) and propose a general algorithm for finding connected components in undirected graphs. We show that, under certain assumptions, the computational cost of the algorithm becomes independent on the size of the sets of vertices and edges (i.e., the algorithm has a constant computational cost with the number of vertices and edges). In consequence, the cost of generating the set of equations in the flattening stage results independent on the size of the arrays of connectors.

Besides introducing and analyzing the algorithm, we also describe a prototype implementation in GNU Octave \cite{eaton1997gnu}. In addition, we analyze three examples (incuding a multidimensional one) showing the efficiency of the novel procedure. 

The paper is organized as follows. After this introduction we briefly present a problem that motivates the work. Then, Section~\ref{sec:background} introduces some concepts and previous works that are used as the basis of the main results, presented in Section~\ref{sec:main}. The prototype implementation of the algorithm is described in Section~\ref{sec:implem} and its usage for flattening connections is discussed in Section~\ref{sec:flattening}. Finally, Section~\ref{sec:examples} introduces some examples and Section~\ref{sec:conclusions} concludes the article.

\subsection{Motivation}
This work was motivated by a problem that appears in Modelica compilers. Modelica models can be represented by the coupling of several sub-models where the coupling is usually made using \emph{connectors}. That way, the equations representing the structure of the circuit of Figure~\ref{fig:rc_network} can be represented by the piece of code in Listing~\ref{list:modelicaCode}.

\begin{figure}[h]
 \centering
 \includegraphics[width=8cm]{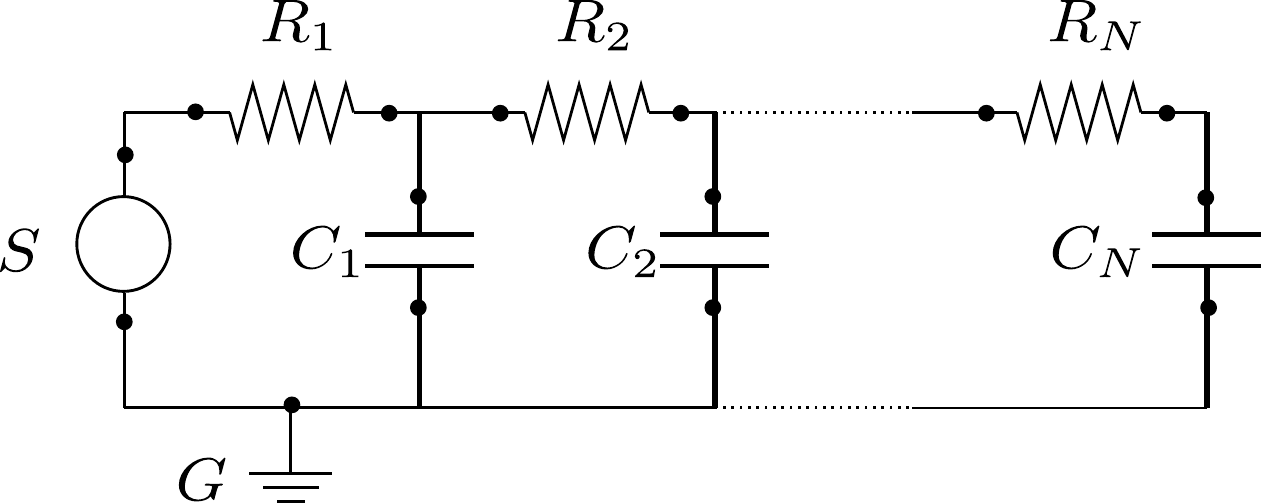}
 \caption{RC network}
 \label{fig:rc_network}
\end{figure}

\begin{scriptsize}
\begin{lstlisting}[language=Modelica,frame=single,caption=Modelica connections,label=list:modelicaCode]
connect(S.p,R[1].p);
connect(S.n,G.p);
for i in 1:N-1 loop
  connect(R[i].n, R[i+1].p);
end for;
for i in 1:N loop
  connect(C[i].p, R[i].n);
  connect(C[i].n, G.p);
end for;
\end{lstlisting}
 \end{scriptsize}

The connectors (\texttt{S.p}, \texttt{S.n}, etc) have two types of variables: \emph{effort} variables that are equal to each other after being connected and \emph{flow} variables whose sum is zero for all connected connectors. Thus, the resulting equations for the structure of Listing~\ref{list:modelicaCode} would be that of Listing~\ref{list:modelicaEquations}

\begin{scriptsize}
\begin{lstlisting}[language=Modelica,frame=single,caption=Modelica connections,label=list:modelicaEquations]
S.p.effort=R[1].p.effort;
S.p.flow+R[1].p.flow=0;
S.n.effort=G.p.effort;
S.n.flow+G.p.flow+sum(C.n.flow)=0;
for i in 1:N-1 loop
  R[i].n.effort=R[i+1].p.effort;
  C[i].p.effort=R[i].n.effort;
  R[i].n.flow+R[i+1].p.flow+C[i].p.flow=0;
end for;
C[N].p.effort=R[N].n.effort;
R[N].n.flow+C[N].p.flow=0;
\end{lstlisting}
\end{scriptsize}

The translation from connections to equations requires finding connected components in a graph where the vertices are the connectors (\texttt{S.p}, \texttt{S.n}, etc.) and the edges are defined by the presence of \texttt{connect} statements between the corresponding connectors.

Modelica compilers solve this problem by first expanding the \texttt{for} statements and the arrays of connectors and then finding the connected components and producing the equations as part of a process known as \texttt{flattening}. The result of this process in a model like that of Listing~\ref{list:modelicaCode} is a large piece of code without the \texttt{for} statements of Listing~\ref{list:modelicaEquations}. In addition, the cost of producing that code is at least linear with the size of the arrays involved ($N$ in the above example).

When $N$ is large (starting from $10^4$ or $10^5$) the computational costs become huge, and the length of the code produced may become intractable for the successive stages of the compilation process. Thus, we expect that the algorithms developed in this work provide a general solution for this problem as well as for other problems that require a compact and efficient connected components analysis in presence of some repetitive or regular structures.

\section{Background} \label{sec:background}
In this section we present some previous results and tools that are used along the rest of the paper.

\subsection{Modelica and Equation-Based Object-Oriented Modeling Languages}
In an effort to unify the different modeling languages used by the different modeling and simulation tools, a consortium of software companies and research groups proposed an open unified object oriented modeling language called \emph{Modelica} \cite{fritzson1998modelica,Fritzson2015}, that in the last two decades was progressively adopted by different modeling and simulation tools.

Modelica allows the representation of continuous time, discrete time, discrete event and hybrid systems. Elementary Modelica models are described by sets of differential and algebraic equations that can be combined with algorithms specifying discrete evolutions. These elementary models can be connected to other models to compose more complex models, facilitating the construction of multi--domain models.  

Modelica models can be built and simulated using different software tools. OpenModelica \cite{fritzson2006openmodelica} is the most complete open source package, while Dymola \cite{BEMH02} and Wolfram System Modeler are the most used commercial tool. There are also some prototype tools oriented to different problems, such as JModelica \cite{aakesson2009jmodelica} (for optimization problems) and ModelicaCC. 

The simulation of Modelica models requires a previous compilation, that transforms the object oriented model description into a piece of code (usually in C language) containing a set of ordinary differential equations (ODE) or differential algebraic equations (DAE) that can be solved by an appropriate ODE or DAE solver. The compilation process is usually divided in several stages: flattening, alias removal, index reduction, equation sorting, and final code generation. 

All Modelica compilers by default expand the arrays and unroll the \texttt{for loop} cycles in the first step of the compilation process. In consequence, in presence of large arrays, the computational cost of the compilation and the length of the produced code can become huge and the tools are unable to simulate systems with more that about $10^5$ state variables. While there are some experimental implementations that avoid expanding and unrolling  \cite{bergero2015efficient,pop2019new}, there is not yet a general solution.

\subsection{Connected Components in Undirected Graphs}
Finding the connected components of an undirected graph is a simple problem for which there are hundreds of algorithms. Linear time algorithms have been known since a long time ago \cite{hopcroft1973algorithm}, and there are also several parallel algorithms that can reduce the costs to logarithmic time. Among them, we shall briefly describe that of \cite{hirschberg1979computing}, which has certain features in common with the algorithm that constitutes the main result of this work. 

This algorithm represents the connected components using a vector $D$ of length $n$ (the number of vertices in the graph) such that $D(i)$ contains the smallest numbered vertex in the connected component to which $i$ belongs. A version of this procedure is described in Algorithm~\ref{alg:hirsch}, where we consider that a graph $G=(V,E)$ is given with a set of vertices $V=\{1,2,\ldots,n\}$ and a set of edges $E=\{e_1,\ldots,e_m\}$ with $e_k=\{i,j\}$ where $i,j\in V$. 

\begin{algorithm}[h]
\begin{algorithmic}[1]
\Function{Connect}{$V,E$} \Comment{All the steps are performed in parallel for all $i\in V$}
 \State $D(i)\gets i$ for all $i\in V$.
 \For{$it_1=1:\log_2(n)$}
   \State $C(i) \gets \min_j \left (D(j) | \{C(i),D(j)\}\in E \wedge D(j)\neq D(i)\right )$, if none then $D(i)$, for all $i\in V$ \label{li:collapse}
   \State $C(i) \gets \min_j \left (C(j) | D(j)=i \wedge C(j)\neq i\right )$, if none then $D(i)$, for all $i\in V$
   \State $D(i)\gets C(i)$ for all $i\in V$.
    \For{$it_2=1:\log_2(n)$}
      \State $C(i)\gets C(C(i))$ for all $i\in V$. \label{li:selfmap}
    \EndFor
   \State $D(i)\gets \min(C(i),D(C(i)))$ for all $i\in V$.    
 \EndFor
 \State \Return{$D$}
\EndFunction
 \end{algorithmic}
 \caption{Connected Components of \cite{hirschberg1979computing}}
 \label{alg:hirsch} 
\end{algorithm}

The details and the explanation of this algorithm is given in \cite{hirschberg1979computing}. The algorithm we shall develop will use a very similar idea to represent the connected components (with a more general idea of the vertex numbering) and we shall also use an auxiliary vector like $C(i)$ with a similar idea for merging the components in step \ref{li:collapse} and applying the map into itself like in step \ref{li:selfmap} until all the members of a component point to the same root vertex. 

\subsection{Set--Based Graphs}
The algorithms presented in this work are based on the use of \emph{Set-Based Graphs} (SB-Graphs), first defined in \cite{ZFK19}. SB-Graphs are regular graphs in which the vertices and edges are grouped in sets allowing sometimes a compact representation. We introduced next the main definitions.

\begin{defn}[Set--Vertex]
  A \emph{Set--Vertex} is a set of vertices $V=\{v_1,v_2,\ldots,v_n\}$.
\end{defn}

\begin{defn}[Set--Edge]
 Given two  Set--Vertices, $V^a$ and $V^b$, with $V^a\cap V^b=\emptyset$, a Set--Edge connecting $V^a$ and $V^b$ is a set of non repeated edges  $E[\{V^a,V^b\}]=\{e_1,e_2,\ldots,e_n\}$ where each edge is a set of two vertices $e_i=\{v^a_k \in V^a,v^b_l \in V^b\}$.
\end{defn}

\begin{defn}[Set--Based Graph]
  A Set--Based Graph is a pair $\G=(\V,\E)$ where
  \begin{itemize}
   \item $\V=\{V^1,\ldots,V^n\}$ is a set of disjoint set--vertices (i.e., $i\neq j \implies V^i\cap V^j=\emptyset$).
   \item $\E=\{E^1,\ldots,E^m\}$ is a set of set--edges connecting set--vertices of $\V$, i.e., $E^i=E[\{V^a,V^b\}]$ with $V_a\in \V$ and $V_b\in \V$. In addition, given two set edges  $E^i, E^j\in \E$ with $i\neq j$, such that $E^i=E[\{V^a,V^b\}]$ and $E^j=E[\{V^c,V^d\}]$, then $V^a\cup V^b \cup V^c \cup V^d\neq V^a\cup V^b$. This is, two different set--edges in $\E$ cannot connect the same set--vertices.
   \end{itemize}
\end{defn}

As in regular graphs, we can define bipartite Set--Based Graph and directed Set--Based Graphs. An algortihm for matching in bipartite Set--Based Graph and an algorithm for finding the strongly connected components of a directed Set--based Graph were recently presented in \cite{ZFK19}. 

An SB-Graph $\G=(\V,\E)$ defines an equivalent regular graph $G=(V,E)$ where $V=\bigcup V^i\in \V$ and $E=\bigcup E^i \in \E$. Thus, a SB--Graph contains the same information than a regular graph. However, SB-Graphs can have a compact representation of that information provided that every set--edge and every set-vertex is defined by \emph{intension}.

\section{Main Results} \label{sec:main}
This section introduces the main result of the article. We first introduce a simple but inefficient algorithm for finding the connected components of regular graphs. Then we show that this algorithm, in the context of Set--Based Graphs, can be implemented using compact operations on some sets and maps leading to computational costs that, under certain circumstances, become independent on the number of vertices and edges.

\subsection{An Inefficient Algorithm for Regular Graphs}
We present first an algorithm for computing the connected components in a regular graph $G=(V,E)$. The proposed algorithm finds a collection of connected components represented in a similar way to that Algorithm~\ref{alg:hirsch}. In particular:
\begin{itemize}
 \item We assume that there exists a total ordering between all individual vertices (they could be represented by integer numbers, by arrays of integer numbers, etc).
 \item Each connected component is represented by one of its vertices $v_k\in V$, which is the smallest vertex of the connected component.
 \item There is a map $\Rmap:V \to V$ such that $\Rmap(v_r)=v_k$ implies that the vertex $v_r \in V$ is part of the connected component represented by $v_k$. 
 \item Since the representative $\Rmap(v_r)$ is the minimum vertex on the connected component, then $\Rmap(v_r)\leq v_r$ for all $v_r\in V$. 
\end{itemize}

Making use of this representation, Algorithm~\ref{alg:main_map} finds the connected components represented by $\Rmap$ of an arbitrary graph $G=(V,E)$.

\begin{algorithm}[h]
\begin{algorithmic}[1]
\Function{Connect}{$V,E$}
  \State{$\Rmap\gets \mathrm{Identity_{map}}:V\to V$} \Comment{All vertices are initially disconnected}
  \State{$I_{\mathrm{old}} \gets \emptyset$} \Comment{Previous image set of $\Rmap$}
  \While{$I_{\mathrm{old}} \neq \Image(\Rmap)$}
    \State{$\tRmap\gets \Rmap$} \Comment{New map of connected components}
    \ForAll{$v_r\in \Image(\Rmap)$} \Comment{Component represented by $v_r$} \label{li:for}
      \If{$\exists \{v_r,v_s\}\in E$} \Comment{$v_r$ is not an isolated vertex}
	\State{$v_k\gets \min(\Rmap(v_b):(\{v_a,v_b\}\in E \wedge \Rmap(v_a)=v_r))$} \Comment{Minimum component connected to the component represented by $v_r$}
	\If{$v_k<v_r$}
	  \State{$\tRmap(v_r)\gets v_k$} \Comment{Connect components represented by $v_r$ and $v_k$}
	  \State{$\tRmap(v_a) \gets \tRmap\circ \tRmap(v_a)=\tRmap(v_r)=v_k$ for all $v_a:\tRmap(v_a)=v_r$} \Comment{All components represented by $v_r$ are now represented by $v_k$} \label{li:selfit}
	\EndIf
      \EndIf
    \EndFor \label{li:endfor} 
    \State{$I_{\mathrm{old}} \gets \Image(\Rmap)$} \Comment{Image of the previously connected components}
    \State{$\Rmap\gets \tRmap$} \Comment{New map of connected components}
  \EndWhile
  \State \Return{$\Rmap$}
\EndFunction  
\end{algorithmic}
\caption{Connected Components} \label{alg:main_map}
\end{algorithm}

The algorithm works as follows. It starts assuming that all vertices are disconnected so they represent their own connected component. Then, it iterates until the image of $\Rmap$ becomes constant, meaning that no further components can be connected. 

During each iteration a new map $\tRmap$ is computed by adding connections between components. For each component represented by $v_r$, the algorithm takes into account all the edges connecting vertices of this component. Among all these edges, it takes the one that connects to certain vertex $v_b$ with the least representative $v_k=\Rmap(v_b)$ (it could happen that $v_k=v_r$ if there is no connection from the component represented by $v_r$ to any component represented by a smaller vertex).  Then, if the representative $v_k$ is smaller than $v_r$, the algorithm connects both components by making $\tRmap(v_r)=v_k$. In that case, it also reconnects all the vertices that were connected to $v_r$ such that they are now connected to $v_k$.

Although it could be easily proved that the procedure is correct, it is possibly one of the less efficient algorithms one can imagine to find connected components in a graph. Its computational cost appears to grow at least quadratically with the number of vertices and edges. However, we shall see next that in the context of Set--Based Graph this algorithm can be implemented in a way that the costs become independent on the size of the different sets involved.

A key feature of the algorithm above that will allow this simplification is that in each iteration $\tRmap$ is computed as a function of the complete map $\Rmap$ and vice-versa. That way, both maps can be entirely computed from each other in simple steps. 

\subsection{Set--Based Graph Algorithm}
The goal of using Set--Based Graph is to exploit the presence of repeating regular structures along the graph, representing the different sets by intension. While the definitions of SB--Graphs do not explicitly establish this, we propose next a simple way of representing the set edges that allows the intensive  treatment of the graph.

Let $E^h$ be a set-edge connecting $V^i$ and $V^j$. We shall characterize this set--edge using two maps that relate the individual edges $e^h_k\in E$ with the vertices it connects $v^i_r=\mathrm{map}^{h,i}(e^h_k)$ and $v^j_s=\mathrm{map}^{h,j}(e^h_k)$. This is, the set edge is compactly defined as
\begin{equation*}
  E^h=\bigcup_k \{v^i_r=\mathrm{map}^{h,i}(e^h_k) , v^j_s=\mathrm{map}^{h,j}(e^h_k)\}.
\end{equation*}
Thus, provided that there is a compact expression for these maps and that the set-vertices are represented by intension, the complete SB--Graph has a compact representation. 

Using this representation of an SB--Graph, the previous algorithm can be reformulated as proposed in Algorithm~\ref{alg:connect_SB}.

\begin{algorithm}[h]
\begin{algorithmic}[1]
\Function{ConnectSBG}{$\V,\E$}
  \State{$V\gets \bigcup V^i\in \V$} \Comment{Set of all vertices}
  \State{$(\Emap^1,\Emap^2) \gets \mathrm{edgeMaps}(\E)$} \Comment{Left and right maps from edges to vertices}
  \State{$\Rmap\gets \mathrm{Identity_{map}}:V\to V$} \Comment{All vertices are initially disconnected}
  \State{$I_{\mathrm{old}} \gets \emptyset$} \Comment{Previous image set of $\Rmap$}
  \While{$I_{\mathrm{old}} \neq \Image(\Rmap)$} \label{li:while}
    \State{$\ERmap^1 \gets \Rmap\circ \Emap^1$} \Comment{Left map from edges to connected components}
    \State{$\ERmap^2 \gets \Rmap\circ \Emap^2$} \Comment{Right map from edges to connected components}
    \State{$\tRmap^1\gets \mathrm{minAdjMap}(\ERmap^1,\ERmap^2)$} \Comment{Map from components to least components via $\Emap^2$}
    \State{$\tRmap^2\gets \mathrm{minAdjMap}(\ERmap^2,\ERmap^1)$} \Comment{Map from components to least components via $\Emap^1$}
    \State{$\tRmap\gets \min(\Rmap,\tRmap^1,\tRmap^2)$} \Comment{Map from components to least components}
    \State{$I_{\mathrm{old}} \gets \Image(\Rmap)$} \Comment{Image of the previously connected components}
    \State{$\Rmap\gets (\tRmap)^{\infty}$} \Comment{New map of connected components}
  \EndWhile \label{li:endwhile}
  \State{\Return{$\Rmap$} }
\EndFunction  
\end{algorithmic}
\caption{Connected Components with SB--Graphs} \label{alg:connect_SB}
\end{algorithm}

In this new algorithm, we made use of the following functions and notation:
\begin{itemize}
 \item Function $\mathrm{edgeMaps}(\E)$ returns two maps: a map of left connections $\Emap^1:E\to V$ and a map of right connections $\Emap^2:E\to V$, defined as follows. For each set--edge $E^h\in\E$ connecting set vertices $V^i$, $V^j$, the maps $\Emap^{1,2}$ satisfy
 \begin{equation*}
  \begin{split}
    \Emap^1(e_k^h)&=map^{h,i}(e_k^h) \forall e_k^h \in E^h\\   
    \Emap^2(e_k^h)&=map^{h,j}(e_k^h) \forall e_k^h \in E^h   
  \end{split}
 \end{equation*}
  Notice that for each set edge, there are two possible definitions of $\Emap^1$ and $\Emap^2$, according to which one is associated with $i$ and which one with $j$ (the set--edges are non--directed).
  
  \item Function $\mathrm{minAdjMap}(\map_1,\map_2)$ computes a map $\map_3$ such that
  \begin{equation} \label{eq:minAdjMap}
     \map_3(v)=\min(\map_2(e):\map_1(e)=v) 
  \end{equation}
  
  In the context of this algorithm, $v$ is a representative vertex and $e$ is an edge. Thus, for all edges such that $\map_1(e)=v$, the function takes the one for which $\map_2(e)$ is minimum and defines $\map_3(v)=\map_2(e)$. That way, $\map_3(v)$ is the least representative vertex connected via $\map_2$ to a vertex represented by $v$. 
  
  In the algorithm, the function is invoked twice with the inverted arguments in order to find the least representative connected to a component via both maps.  
  
  \item The notation $(\tRmap)^{\infty}$ is the result of applying $\tRmap$ on itself until arriving to a fixed point. 
\end{itemize}

The algorithm is almost identical to the previous one, except that the iteration of $\tRmap$ on itself (step~\ref{li:selfit} in Algorithm~\ref{alg:main_map}) is now performed at the end of the cycle. The convergence of this new iteration is ensured by the fact that $\tRmap$  is always less or equal than the identity map and that its domain is finite ($V$).

\subsection{About the Computational Costs}
We shall see in the next section that, under certain assumptions on the definition of the maps, all the steps involved in this new algorithm can be computed by intension (including the infinite iteration of $\tRmap$ on itself). Then, the computational cost of each iteration of the algorithm (steps \ref{li:while}--\ref{li:endwhile}) becomes independent on the size of the sets. 

Regarding the number of iterations that are actually needed until all components are connected, the following result establishes an upper bound.

\begin{lemma} \label{lem:itbound}
 The numbers of iterations required to find all connected components is at most $2\log_2(N)$, where $N$ is the number of edges in the largest connected componet. 
\end{lemma}
\begin{proof}
 Suppose that after certain number of iterations $k$, a component represented by $v_r$ contains one or more connections to other components represented by $v_{s_1}$, $v_{s_2}$, etc. Suppose also that during the next iteration the component represented by $v_r$ is not connected to any of those components.
 
 If that occurs is because $v_r<v_{s_i}$ (otherwise it would be connected to the component represented by the minimum $v_{s_i}$). In addition, the components represented by $v_{s_i}$ will be connected in that iteration to some components represented by $v_{t_j}<v_r$ (otherwise, they would be connected to the component represented by $v_r$). Then, in the following iteration, $v_r$ will have connections to components represented by $v_{t_j}<v_r$ and it will be connected to the least $v_{t_j}$.
 
 Thus, every component containing connections to other components is always connected after a maximum of two iterations. It means that after two iterations the number of different components that will be part of the same connected component is reduced at least to the half and they will be reduced to a single component after at most $2\log_2(N)$ iterations.
\end{proof}

This lemma tells that the number of iterations (and so the computational costs) of the algorithm may actually depend on the size of the sets. However, in several cases it does not:

\begin{enumerate}
 \item When the structure is such that each connected component can only have a bounded number of vertices (independently of the size of the set-vertices). 
 \item When the latter condition is not accomplished by some connected components, but each connected component can be split in two components: the first one verifying the previous condition and the second one having all its vertices disconnected among them but connected to some vertices of the first component.
 \item When the second component of the previous case does have connections among its vertices, but the connections follow an order: A connection between ($v_{r_1}-v_{r_2}-v_{r_3}-\ldots-v_{r_p}$), implies that $v_{r_1}<v_{r_2}<v_{r_3}<\ldots<v_{r_p}$. 
\end{enumerate}

The independence of the computational costs with the size of the sets in the first case is ensured by Lemma~\ref{lem:itbound}. 

In the second case, the fact that the \emph{large} set of edges has only connections to the small set of edges implies that in at most two iterations the edges of the large set will be connected to the edges of the small set (the reason for that can be found in the proof of Lemma~\ref{lem:itbound}). After that, the number of components is reduced to a quantity that is independent on the size of the sets and so is the number of additional iterations.

In the third case, each connection of the form $v_{r_1}-v_{r_2}-v_{r_3}-\ldots-v_{r_p}$ with $v_{r_1}<v_{r_2}<v_{r_3}<\ldots<v_{r_p}$ produces that all the components get connected in a single iteration of the algorithm (unless they are first connected to the small set of components). Then, in either situation, the case reduces to  the situation analyzed in the previous case.

In conclusion, the only situation in which a large number of iterations would be required is under the presence of a large connected component resulting from a large non--ordered set of connections. Yet, that would be only possible when the maps that define the set edges have some irregular definition.  

\section{Implementation} \label{sec:implem}
Algorithm \ref{alg:connect_SB} was implemented in a prototype library of Octave for Set--Based Graphs. The library defines four basic classes: \texttt{Interval}, \texttt{Set}, \texttt{Map}, and \texttt{SBGraph} and different operations involving them. We describe next their main features.

\subsection{Intervals}
A unidimensional interval is represented by three natural numbers: \texttt{Interval.start}, \texttt{Interval.step}, and \texttt{Interval.end}. For instance, the sequence $[3,5,7,\ldots,199]$ is compactly represented by \texttt{start}=3,  \texttt{step}=2, and \texttt{end}=199 (we shall simply denote it by $[3:2:199]$).

A general interval of dimension $d$ is represented by three arrays of length $d$: \texttt{Interval.start}($1:d$), \texttt{Interval.step}$(1:d)$, and \texttt{Interval.end}$(1:d)$. For instance, the sequence
\begin{equation*}
 [(1;1), (1;2), \ldots, (1;100), (4;1), (4;2), \ldots, (4;100), \ldots, (1000;1), (1000;2), \ldots, (1000;100)]
\end{equation*}
is represented by $\texttt{start}(1)=1$, $\texttt{step}(1)=1$, $\texttt{end}(1)=100$, $\texttt{start}(2)=1$, $\texttt{step}(2)=3$, $\texttt{end}(2)=1000$. We shall denote it by $[1:1:100]\times [1:3:1000]$. 

On these intervals we defined some basic functions and operations used by the higher level class that defines sets.

\subsection{Sets}
A set is defined as an array of disjoint intervals of the same dimension. This is, $\texttt{Set.Interval}(1)$ contains the first interval, $\texttt{Set.Interval}(2)$ contains the second interval, etc. For instance, the set
\begin{equation*}
 S=\{2,4,6,\ldots,100\} \cup \{101,102,\ldots,200\}
\end{equation*}
is represented by an array of two intervals: $[2:2:100]$ and $[101:1:200]$ and we shall denoted it as $S=\{[2:2:100]\} \cup \{[101:1:200]\}$.

On the set class, we defined some functions and operators, including the basic operations \texttt{setUnion}, \texttt{setIntersection}, and \texttt{setMinus}. All the operations are computed by \emph{intension} using only the \texttt{start}, \texttt{step} and \texttt{end} values of the underlying intervals, and the result is another set represented by intervals. That way, the cost of the operations is independent on the size of the intervals involved.

\subsection{Maps}
A one dimensional linear map is defined by two rational numbers: \texttt{linearMap.gain} (which cannot be negative) and  \texttt{linearMap.offset}. Similarly, a general $d$--dimensional linear map is defined by two arrays of length $d$ \texttt{linearMap.gain}$(1:d)$, and \texttt{linearMap.offset}$(1:d)$.

A \texttt{Map} is then defined by an array of disjoint sets \texttt{Map.domain}$(1:M)$ and an array of linear maps \texttt{Map.linearMap}$(1:M)$, where all the sets and linear maps have the same dimension. For instance, a map like
\begin{equation*}
 i=\begin{cases}
    j+3 & \text{for~} j \in \{1,2,\ldots,100\}\\
    100 & \text{for~} j \in \{101,103,\ldots,199\}\\
    j/2 & \text{for~} j \in \{102,104,\ldots,200\}\\   
   \end{cases}
\end{equation*}
is defined by 
\begin{itemize}
 \item \texttt{Map.domain(1)}=$\{1:1:100\}$, \texttt{Map.linearMap(1).gain}=1, \texttt{Map.linearMap(1).offset}=3
 \item \texttt{Map.domain(2)}=$\{101:2:199\}$, \texttt{Map.linearMap(2).gain}=0, \texttt{Map.linearMap(2).offset}=100
 \item \texttt{Map.domain(3)}=$\{102:2:200\}$, \texttt{Map.linearMap(1).gain}=1/2, \texttt{Map.linearMap(1).offset}=0
\end{itemize}
A restriction in the definition of a map is that every domain and its correspondent linear map must be such that the resulting image in each dimension is composed by natural numbers. Thus, when a gain is not an integer number, the corresponding domain and offset cannot be arbitrary. Otherwise, if the gain is integer, the offset must be integer too.

On these maps we also implemented several functions and operators. Among them, we mention the following ones:

\begin{itemize}
 \item \texttt{imageMap} computes the set that is the image of a given set through a given map. Similarly, \texttt{preImageMap} computes the preimage set.
 \item \texttt{compMaps} computes the new map that results from composing two maps ($\map_3=\map_1\circ \map_2$).
 \item \texttt{minMap} computes the minimum map between two maps, i.e., $\map_3(v)=\min(\map_1(v),\map_2(v))$, which can result equal to $\map_1$ in some subdomain, and equal to $\map_2$ in the remaining subdomain. 

 This function requires establishing an ordering between the elements. For one dimensional sets the ordering is that of the natural numbers. For higher dimensional sets, the order between two elements is established at the first dimension in which they differ. This is, we say that $v<w$ if $v_1<w_1$ or $v_1=w_1 \wedge v_2<w_2$, etc.

 \item \texttt{minAdjMap}: Given two maps $\map_1$ and $\map_2$ with the same domain, this function computes a new map $\map_3$ according to Eq.\eqref{eq:minAdjMap}. 
   The computation of the new function is based on the following observation:
   \begin{itemize}
    \item If $\map_1$ is bijective, then $\map_3$ can be computed as $\map_2 \circ \map_1^{-1}$.
    \item If $\map_1$ is constant, then $\map_3$ can be computed as $\map_3(v)=\min(\map_2(e))$ for all $e$ in the domain of the maps.
   \end{itemize}
 Then, the function is implemented computing on each sub-domain and on each dimension of $\map_1$ according to the previous observation. 
 
 \item \texttt{mapInf}: Consider a map $\map_1$ with the following restrictions: 
 \begin{itemize}
  \item All its linear maps have gains (in all the dimensions) that can only take the values 1 and 0.
  \item If a gain is 1, the corresponding offset cannot be positive.
 \end{itemize}
 On this map, this function computes a new map $\map_2$ that is the result of composing $\map_1$ with itself until reaching convergence. The computations are performed without actually iterating on $\map_1$. Instead, it computes the fixed points of the iteration and the maps to those fixed points.
 
 The implementation is based on the following observations:
 \begin{itemize}
   \item A domain where the map has gain 1 and offset 0 remains unchanged after each iteration.
   \item If all domains have gain 0, then the iteration converges after at most $N$ steps where $N$ is the number of domains. 
   \item If a domain has gain 1 and offset -1, then after some iterations of the map it will take a value outside the domain ($\texttt{interval.start}-1$ in fact). Thus we can just replace the gain by 0 and the offset pointing to $\texttt{interval.start}-1$.
   \item If a domain has gain 1 and offset -2, we shall have two arrival points after leaving the domain. So we can split the interval in two intervals with gain 0 and different offset. For larger negative offset values the idea is the same.
 \end{itemize}
\end{itemize}

\subsection{Set--Based Graphs}
Set--Based Graphs are represented by an array of sets \texttt{SBG.setVertex}$(1:n)$ containing the set vertices and an array of set edges \texttt{SBG.setEdge}$(1:m)$.  

Every set-edge contains two integer numbers \texttt{SE.index1}, \texttt{SE.index2} and two maps, \texttt{SE.map1}, \texttt{SE.map2}, with identical domain. The integer numbers represent the position of the set--vertices that are connected by the set edge, and the maps represent the connections between individual vertices. For instance, a set--edge with \texttt{index1}=3 and \texttt{index2}=5 connects the set vertices \texttt{SBG.setVertex}$(3)$ with \texttt{SBG.setVertex}$(5)$. Then, given $h\in$\texttt{SE.map1.domain}, the $h$--th edge of the set--edge connects the vertices \texttt{SE.map1}$(h)$ with \texttt{SE.map2}$(h)$.

On this class, we implemented the function \texttt{connectComp} that computes the connected components of a given SB-Graph. This function returns a map $\Rmap$ as explained in Section~\ref{sec:main}.

\subsection{Implementation Restrictions} \label{sec:restrict}
While Algorithm~\ref{alg:main_map} is general, the implementation described above imposes the following restrictions on the set--based graphs:
\begin{enumerate}
 \item Every individual vertex is represented by an array of natural numbers of dimension $d$.
 \item Every set-vertex is a union of a finite number of intervals of dimension $d$. Every interval in each dimension is defined by three natural numbers: \emph{start}, \emph{step}, and \emph{end}.
 \item The maps that define the set edges $\map^{h,i}:\N^d\to \N^d$ are \emph{piecewise linear}. Each map has a finite number of domains with a corresponding linear affine function. In every domain, the function acting in each dimension is characterized by two rational numbers: the \emph{gain} and the \emph{offset}. 
 \item The implementation of the \texttt{mapInf} function imposes a further restriction to the maps: In a given domain and dimension, if $\map^{h,i}$ and $\map^{h,j}$ have both nonzero gains, then the gains must be the same. Otherwise, function \texttt{minAdjMap} might return a map with some gain that is not $1$ or $0$ and, if that map turns to be less than the identity, then \texttt{mapInf} cannot be applied. 
\end{enumerate}

The last restriction can be easily avoided with a more general implementation of \texttt{mapInf} considering gains different from 1 and 0. 

\section{Application to Connection Flattening} \label{sec:flattening}
In this section we analyze the use of the proposed algorithm in the context of replacing connections by equations in object oriented models.

\subsection{Code Generation}
The original motivation of this work was that of automatically obtaining a code like that of Listing~\ref{list:modelicaEquations} given a set of connections like those of Listing~\ref{list:modelicaCode}. For that goal, we propose the following procedure:
\begin{enumerate}
 \item Build a SB Graph: 
 \begin{itemize}
  \item  Associate a set-vertex to each array of connectors. For the example of Listing~\ref{list:modelicaCode} the set vertices are $S.p$, $S.n$, $G.p$, $R[1:N].p$, $R[1:N].n$, $C[1:N].p$, and $R[1:N].n$.
  \item  Associate a set-edge to each set of connections between every pair of set vertices. In the example some set edges would be 
    \begin{itemize}
      \item $E^1=E^1[S.p,R.p]$, characterized by maps $\map^1_1(e^1_1)=S.p$ and $\map^1_2(e^1_1)=R[1].p$. 
      \item $E^2=E^2[R.n,R.p]$, characterized by maps $\map^2_1(e^2_i)=R[i].n$ and $\map^2_2(e^2_i)=R[i+1].p$ for $i=1,\ldots,N-1$. 
      \item $E^3=E^3[C.n,G.p]$, characterized by maps $\map^3_1(e^3_i)=C[i].n$ and $\map^3_2(e^3_i)=G.p$  for $i=1,\ldots,N$. 
    \end{itemize}
 \end{itemize}
 \item Find the connected components using Algorithm~\ref{alg:connect_SB}.
 \item Given the map $\Rmap$ representing the sets of connected components, write the corresponding equations. 
\end{enumerate}

The last step first splits the domain and image of $\Rmap$ into \emph{atomic} sets, i.e., sets containing a single intervals. That way, the sets can be traversed in the resulting code using \texttt{for} statements. Then, the procedure uses the facts that the image of $\Rmap$ are the representatives of the connected components and that the preimage of each atomic set of the image contains the corresponding connected components. Since the preimage is also split into atomic sets, it can be also traversed using \texttt{for} statements in the resulting code. Then, once the code for traversing the connected components is written, it is simple to add the appropriate code for the effort and flow variables.

\subsection{Analysis of the Restrictions}
The restrictions described in Sec.\ref{sec:restrict} about the implementation and the conditions enumerated after Lemma~\ref{lem:itbound} establishes the circumstances under which the algorithm effectively achieves a constant cost with respect to the number of vertices and edges. While these conditions may be quite restrictive in general, in the context of replacing connections by equations in object oriented models they are almost invariantly satisfied:
\begin{itemize}
 \item The connectors in a model are always instantiated as scalar or arrays with different dimensions. We can represent all of them using arrays of vertices with the maximum dimension found. That way the first two restrictions of Sec.\ref{sec:restrict} are always satisfied.
 \item The third restriction is satisfied provided that:
   \begin{itemize}
     \item In presence of nested \texttt{for loop} statements, the interval of the iterators are independent on each other. This is, we cannot write \texttt{for i in 1:N loop; for j in 1:i loop} since in that case the domain of the maps defining the set edges would not be an interval.
      \item The connections have linear affine operations with each index. This is, we can only have expressions like \texttt{connect(v[a*i+b, c*j+d], w[e*i+f, g*j+h])} where \texttt{i} and \texttt{j} are the nested iterators and \texttt{a, b, c, d, e, f, g, h} are rational constants.
   \end{itemize}
   \item The fourth restriction is satisfied provided that \texttt{a} and \texttt{e} in the previous item are different only if one of them is zero (and the same for \texttt{c} and \texttt{g}).
\end{itemize}

Regarding the conditions listed after Lemma~\ref{lem:itbound} under which the algorithm performs a limited number of iterations, they are automatically satisfied under the assumption that the maps are piecewise linear since in that case any large set of connected connectors will keep a strict ordering.

\section{Examples and Results} \label{sec:examples}
We introduce three examples where we applied the presented algortihm using the Octave implementation described in Section~\ref{sec:implem}. In all cases, the experiments were run on laptop with an Intel i3 core processor running Ubuntu OS.

\subsection{Simple RC Network}
We consider first the example of Listing~\ref{list:modelicaCode} with $N=1000$. The vertices $S.p$, $S.n$, and $G.p$ are represented by numbers $1$, $2$, and $3$, respectively. The vertices $R[1:1000].p$ are represented by numbers $1001$ to $2000$, and vertices $R[1:1000].n$ by numbers $2001$ to $3000$. Similarly, $C[1:1000].p$ are represented by numbers $3001$ to $4000$, and $C[1:1000].n$ are represented by numbers $4001$ to $5000$.

Using Algorithm~\ref{alg:connect_SB}, the map $\Rmap$ results as follows:
\begin{equation*}
 \Rmap(v)=
 \begin{cases}
v  &  \text{if } v \in\{1002:1:2000\}\\
v -1000 & \text{if } v \in\{4000:1:4000\}\\
v -1999 & \text{if } v \in\{3001:1:3999\}\\
 2 & \text{if } v \in\{3:1:3\}\\
 1 & \text{if } v \in\{1001:1:1001\}\\
 2 & \text{if } v \in\{4001:1:5000\}\\
v  & \text{if } v \in\{3000:1:3000\}\\
v -999 & \text{if } v \in\{2001:1:2999\}\\
v  & \text{if } v \in\{2:1:2\}\\
v  & \text{if } v \in\{1:1:1\}\\
 \end{cases}
\end{equation*}
which can be easily verified to be correct. The representative of the connected components are $S.p$ (represented by number $1$), $S.n$ (represented by number $2$) , $R[2:1000].p$ (represented by numbers $1002$ to $2000$), and $R[1000].n$ (represented by number $3000$).

Octave reports $2.42$ seconds to compute the connected components. The algorithm finishes after only one iteration. In order to check that the computation time was independent on $N$ we repeated the calculations for $N=10,000$, and $N=1,000,000$ and the three cases took almost exactly the same time. It is worth mentioning that the Octave is an interpreter, so the time of $2.42$ seconds would be noticeably reduced on a compiled implementation. 

We also implemented in Octave a simple automatic code generator for connected components. In this example, the generated code is shown in Listing~\ref{list:finalEquations}.

\begin{scriptsize}
\begin{lstlisting}[language=Modelica,frame=single,caption=Generated Equations,label=list:finalEquations]
 for i in {[1001:1:1001]}
  effort(i) = effort(1)
end
for i in {[1:1:1]}
  flow(i) + flow(i+1000) = 0
end
for i in {[3:1:3]}
  effort(i) = effort(2)
end
for i in {[4001:1:5000]}
  effort(i) = effort(2)
end
for i in {[2:1:2]}
  flow(i) + flow(i+1) + sum(flow(i1), for i1 in [4001:1:5000]) = 0
end
for i in {[2001:1:2999]}
  effort(i) = effort(i-999)
end
for i in {[3001:1:3999]}
  effort(i) = effort(i-1999)
end
for i in {[1002:1:2000]}
  flow(i) + flow(i+999) + flow(i+1999) = 0
end
for i in {[4000:1:4000]}
  effort(i) = effort(3000)
end
for i in {[3000:1:3000]}
  flow(i) + flow(i+1000) = 0
end
\end{lstlisting} 
\end{scriptsize}

\subsection{RC Network with Recursive Connection}
For the same system of Figure~\ref{fig:rc_network}, we changed the connections as follows:
\begin{scriptsize}
\begin{lstlisting}[language=Modelica,frame=single,caption=Modelica connections,label=list:modelicaCode2]
connect(S.p,R[1].p);
connect(S.n,G.p);
connect(C[1].n,G.p);
for i in 1:N-1 loop
  connect(R[i].n, R[i+1].p);
  connect(C[i+1].n, C[i].n); //recursive connection
end for;
for i in 1:N loop
  connect(C[i].p, R[i].n);
end for;
\end{lstlisting}
\end{scriptsize}
In this case, the algorithm finds the following map of connected components:
\begin{equation*}
 \Rmap(v)=
 \begin{cases}
 2 &\text{if } v \in\{5000:1:5000\}\\
v  &\text{if } v \in\{1002:1:2000\}\\
v -1000 &\text{if } v \in\{4000:1:4000\}\\
v -1999 &\text{if } v \in\{3001:1:3999\}\\
 2 &\text{if } v \in\{3:1:3\}\\
 1 &\text{if } v \in\{1001:1:1001\}\\
 2 &\text{if } v \in\{4001:1:4001\}\\
v  &\text{if } v \in\{3000:1:3000\}\\
v -999 &\text{if } v \in\{2001:1:2999\}\\
v  &\text{if } v \in\{2:1:2\}\\
v  &\text{if } v \in\{1:1:1\}\\
 2 &\text{if } v \in\{4002:1:4999\}\\
 \end{cases}
\end{equation*}
The map is exactly the same as before, but it is now more partitioned in the domain $[4001:5000]$. The presence of the recursive connection on $C.n$ is solved in a single step by the computation of \texttt{mapInf} function.

The time taken by the algorithm in this case is $3.28$ seconds (reported by Octave), and, as before, the algorithm finishes after completing one iteration. The larger time can be explained by the fact that the number of maps is larger than before. 

\subsection{A Two-Dimensional Network}
This example consists of a 2D network formed by $N \times M$ cells with 4 connectors each (left, right, up and down connectors), a ground component with one connector and a source component with two connectors. The network is connected as it is shown in Figure~\ref{fig:2D_network} and expressed in Listing~\ref{list:modelicaCode2D}.

\begin{figure}[h]
 \centering
 \includegraphics[width=7cm]{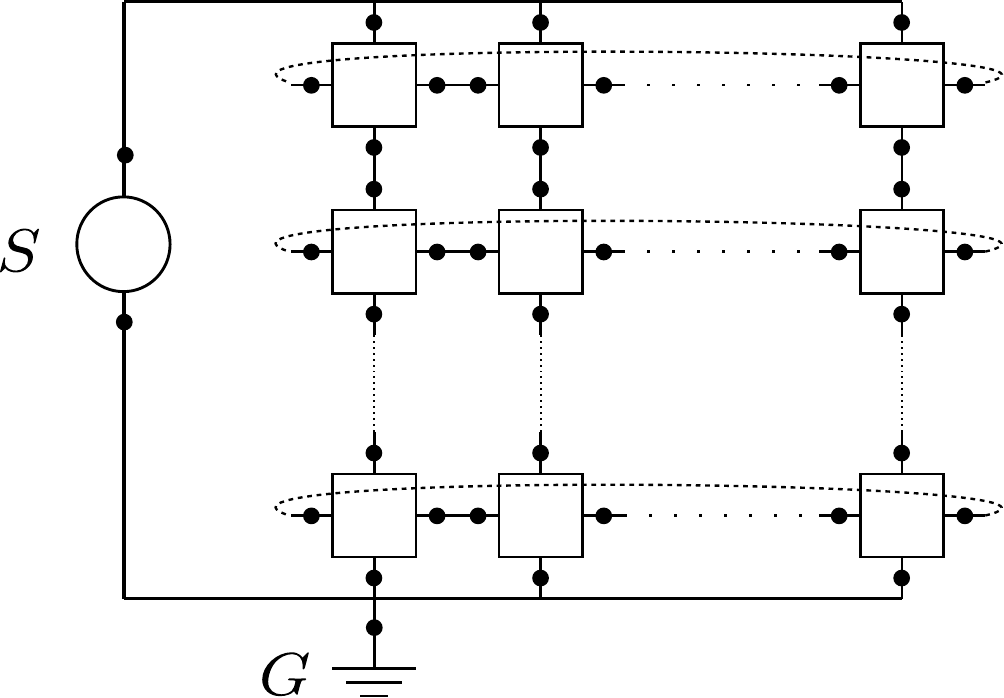}
 \caption{2D Network}
 \label{fig:2D_network}
\end{figure}

\begin{scriptsize}
\begin{lstlisting}[language=Modelica,frame=single,caption=Modelica connections,label=list:modelicaCode2D]
for i in 1:N-1,j in 1:M-1 loop
  connect(Cell[i,j].r, Cell[i,j+1].l);
  connect(Cell[i,j].d, Cell[i+1,j].u);
end for;
for i in 1:N loop
  connect(Cell[i,M].r, Cell[i,1].l);
end for;
for j in 1:M loop
  connect(Cell[1,j].u,S.p);
  connect(Cell[N,j].d,S.n);
end for;
\end{lstlisting}
\end{scriptsize}

In this case, each vertex is represented by two numbers: $S.p$, $S.n$, and $G.p$ are $[1,1]$, $[2,2]$, and $[3,3]$. Cell$[i,j]$.left is represented by $[N+i,M+j]$. Similarly, Cell$[i,j]$.right, Cell$[i,j]$.up, and Cell$[i,j]$.left are represented by $[2 N+i,2 M+j]$, $[3 N+i,3 M+j]$, and $[4 N+i,4 M+j]$ respectively. 

Taking $N=1000$ and $M=100$, for instance, the algorithm finds the following map of connected components:

\begin{equation*}
 \Rmap(v)=
 \begin{cases}
[2;2] & \text{if } v \in \{[3:1:3]\times[3:1:3]\}\\
v  & \text{if } v \in \{[1:1:1]\times[1:1:1]\}\\
v +[-1000;-199] & \text{if } v \in \{[2001:1:3000]\times[300:1:300]\}\\
v +[-999;-100] & \text{if } v \in \{[4001:1:4999]\times[401:1:500]\}\\
v +[-1000;-99] & \text{if } v \in \{[2001:1:3000]\times[201:1:299]\}\\
v  & \text{if } v \in \{[2:1:2]\times[2:1:2]\}\\
[2;2] & \text{if } v \in \{[5000:1:5000]\times[401:1:500]\}\\
[1;1] & \text{if } v \in \{[3001:1:3001]\times[301:1:400]\}\\
v  & \text{if } v \in \{[1001:1:2000]\times[101:1:101]\}\\
v  & \text{if } v \in \{[3002:1:4000]\times[301:1:400]\}\\
v  & \text{if } v \in \{[1001:1:2000]\times[102:1:200]\}\\  
 \end{cases}
\end{equation*}
that can be also verified to be correct. The time reported by Octave is $4.14$ seconds and it is again independent on $N$ and $M$. The code produced is listed below.

\begin{scriptsize}
\begin{lstlisting}[language=Modelica,frame=single,caption=Generated Equations for 2D Network,label=list:finalEquations2D]
for i,j in {[3001:1:3001]x[301:1:400]}
  effort(i,j) = effort(1,1)
end
for i,j in {[1:1:1]x[1:1:1]}
  flow(i,j) + sum(flow(i+3000,j1), for j1 in [301:1:400]) = 0
end
for i,j in {[3:1:3]x[3:1:3]}
  effort(i,j) = effort(2,2)
end
for i,j in {[5000:1:5000]x[401:1:500]}
  effort(i,j) = effort(2,2)
end
for i,j in {[2:1:2]x[2:1:2]}
  flow(i,j) + flow(i+1,j+1) + sum(flow(i+4998,j1), for j1 in [401:1:500]) = 0
end
for i,j in {[2001:1:3000]x[300:1:300]}
  effort(i,j) = effort(i-1000,101)
end
for i,j in {[1001:1:2000]x[101:1:101]}
  flow(i,j) + flow(i+1000,j+199) = 0
end
for i,j in {[2001:1:3000]x[201:1:299]}
  effort(i,j) = effort(i-1000,j-99)
end
for i,j in {[1001:1:2000]x[102:1:200]}
  flow(i,j) + flow(i+1000,j+99) = 0
end
for i,j in {[4001:1:4999]x[401:1:500]}
  effort(i,j) = effort(i-999,j-100)
end
for i,j in {[3002:1:4000]x[301:1:400]}
  flow(i,j) + flow(i+999,j+100) = 0
end
\end{lstlisting} 
\end{scriptsize}

\section{Conclusions and Future Research} \label{sec:conclusions}
We presented a novel algorithm for finding connected components in undirected graph that, under certain regularity assumptions, has constant computational costs with the number of vertices and edges. This is achieved using the concept of \emph{Set-Based Graphs} and, to the best of our knowledge, constitutes the first algorithm of this type.

We described also a prototype implementation of the algorithm and its application to connection flattening in object oriented models, a field in which it is very common that the regularity assumptions are accomplished. In addition, we demonstrated the usefulness and the functionality of the algorithm through three examples of large scale graphs, including a two-dimensional case.

We believe this work opens several future lines of work and research. The implementation itself is a simple prototype in a high level interpreted language, so we are currently working on implementing the algorithm in ModelicaCC compiler \cite{bergero2015efficient} in C++ language. In addition, we are also working on developing more algorithms of this type (using SB-Graphs with maps) for other problems related to Modelica compilation: finding maximum matching in bipartite graphs and strongly connected components (directed graphs). These problems were already solved using SB-Graphs in \cite{ZFK19} but the solution was quite complicated and not as general as the one found here using maps for representing set-edges. Another related problem that we are trying to solve using SB-Graphs is that of producing the code for computing the sparse Jacobian matrix in large systems of differential algebraic equations.

Besides these new problems, there are several issues related to the algorithm presented here that should be taken into account in the future. Among them, it would be important to establish some bounds on the cost of every step of the algorithm with respect to the number of different linear maps that are used to describe each map. In addition, we need to find less restrictive conditions under which the algorithm actually has a constant cost with respect to the size of the sets.
 
Another important goal is that of implementing these algorithms in a more robust and complete Modelica compiler such as OpenModelica \cite{fritzson2006openmodelica}.

Finally, we believe that this algorithm can be effectively applied in other fields beyond object oriented models. Any problem leading to analysis on a large graph containing some regular connections is in principle a good candidate to be solved using SB-Graphs.

The Octave library containing the algorithm, the functions and the examples presented in this article can be downloaded from 
\url{https://www.fceia.unr.edu.ar/~kofman/files/SBGraphs.zip}.

\section*{Funding}
This work was partially funded by grant PICT--2017 2436 (ANPCYT). 

\section*{References}
\bibliography{references}

\end{document}